\documentclass[final,3p]{elsarticle}

\usepackage{amssymb}
\usepackage{graphicx}
\usepackage{amsmath}
\usepackage{url}

\newtheorem{definition}{Definition}
\newtheorem{lemma}{Lemma}
\newtheorem{theorem}{Theorem}

\newproof{proof}{Proof}

\newcommand{\Nat}{\ensuremath{\mathrm{I\kern-1.5pt N}}}

\def\mi#1{\mathit{#1}}

\def\post#1{\ensuremath{{#1}\kern-.05ex\bullet}\,}

\def\postnet#1#2{\ensuremath{{#2}\!\kern-.05ex\stackrel{#1}{\bullet}}\,}

\def\posto#1{\ensuremath{{#1}\kern-.05ex\circ}\,}

\def\postonet#1#2{\ensuremath{{#2}\!\kern-.05ex\stackrel{#1}{\circ}}\,}
\usepackage{amssymb}
\journal{~}

\begin{document}

\begin{frontmatter}



\title{Passages in Graphs}


\author{W.M.P. van der Aalst}

\address{Department of Mathematics and Computer Science,
Technische Universiteit Eindhoven, The Netherlands.\\
BPM Discipline, Queensland University of Technology,
 GPO Box 2434, Brisbane QLD 4001, Australia.\\
WWW: \texttt{www.vdaalst.com}, E-mail: \texttt{w.m.p.v.d.aalst@tue.nl}}

\begin{abstract}
Directed graphs can be partitioned in so-called \emph{passages}.
A passage $P$ is a set of edges such that any two edges sharing the same initial vertex
or sharing the same terminal vertex are both inside $P$ or are both outside of $P$.
Passages were first identified in the context of process mining
where they are used to successfully decompose process discovery and conformance checking problems.
In this article, we examine the properties of passages.
We will show that passages are closed under set operators such as union, intersection and difference.
Moreover, any passage is composed of so-called minimal passages.
These properties can be exploited when decomposing graph-based analysis and computation problems.
\end{abstract}


\end{frontmatter}

% \linenumbers


\section{Introduction}\label{sec:intro}

Recently, the notion of \emph{passages} was introduced in the context of process mining \cite{aalst-passages-PN2012}.
There it was used to decompose process discovery and conformance checking problems \cite{process-mining-book-2011}.
Any directed graph can be partitioned into a collection of non-overlapping passages.
Analysis can be done per passage and the results can be combined easily, e.g.,
for conformance checking a process model can be decomposed into process fragments using passages
and traces in the event log fit the overall model if and only if they fit all process fragments.

As shown in this article, passages have various elegant problems. Although the notion of passages
is very simple, we could not find this graph notion in existing literature on (directed) graphs \cite{digraphs-book-2004,handbook_graph_theory2004}.
Classical graph partitioning approaches \cite{karpis_graph-partitioning1998,kernighan_graph-partitioning1970} decompose the vertices of a graph rather than the edges,
i.e., the goal there is to decompose the graph in smaller components of similar size that have few connecting edges.
Some of these notions have been extended to vertex-cut graph partitioning \cite{feige-vertex-graph-partitioning2005,kim-vertex-graph-partitioning2012}.
However, these existing notions are not applicable in our problem setting where components need to behave synchronously
and splits and joins cannot be partitioned.
We use passages to \emph{decompose a graph into sets of edges such that all
edges sharing an initial vertex or terminal vertex are in the same set}.
To the best of our knowledge, the notion of passages has not been studied before.
However, we believe that this notion can be applied in various domains (other than process mining).
Therefore, we elaborate on the foundational properties of passages.

The remainder is organized as follows.
In Section~\ref{sec:defpas} we define the notion of passages, provide alternative characterizations,
and discuss elementary properties.
Section~\ref{sec:paspart} shows that any graph can be partitioned into passages and that any passage is composed of so-called minimal passages.
Section~\ref{sec:pasgraph} introduces passage graphs visualizing the relations between passages.
Graphs may be partitioned in different ways. Therefore, Section~\ref{sec:quality} discusses the quality of passage partitionings.
Section~\ref{sec:concl} concludes this article.

\section{Defining Passages}\label{sec:defpas}
Passages are defined on directed graphs, simply referred to as graphs.
\begin{definition}[Graph]
A (directed) graph is a pair $G=(V,E)$ composed of a set of vertices $V$ and a set of edges $E\subseteq V \times V$.
\end{definition}
\begin{figure}[htb]
\centerline{\includegraphics[width=9cm]{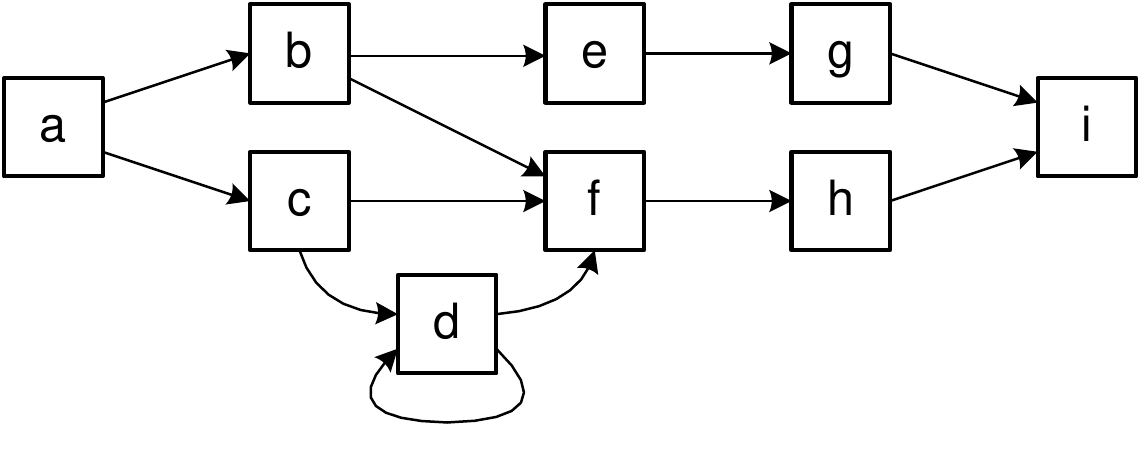}}
\caption{Graph $G_1$ with 9 vertices, 12 edges, and 32 passages.}
\label{fig-passage-init}
\end{figure}
A passage is a set of edges such that any two edges sharing the same initial vertex (tail)
or sharing the same terminal vertex (head) are both inside or both outside of the passage.
For example, $\{(a,b),(a,c)\}$ is a passage in graph $G_1$ shown in Figure~\ref{fig-passage-init}
because there are no other edges having $a$ as initial vertex or $b$ or $c$ as terminal vertex.
\begin{definition}[Passage]
Let $G=(V,E)$ be a graph. $P\subseteq E$ is a passage if for any $(x,y)\in P$
and $\{(x,y'),(x',y)\} \subseteq E$: $\{(x,y'),(x',y)\} \subseteq P$.
$\mi{pas}(G)$ is the set of all passages of $G$.
\end{definition}
Figure~\ref{fig-passage-split} shows 7 of the 32 passages of graph $G_1$ shown in Figure~\ref{fig-passage-init}.
$P_2 = \{(b,e),(b,f),(c,f),(c,d),(d,d),(d,f)\}$ is a passage as there are no other edges having $b$, $c$, or $d$ as initial vertex
or $d$, $e$, or $f$ as terminal vertex. Figure~\ref{fig-passage-split} does not show the two trivial passages: $\emptyset$ (no edges) and $E$ (all edges).
\begin{figure}[htb]
\centerline{\includegraphics[width=10.5cm]{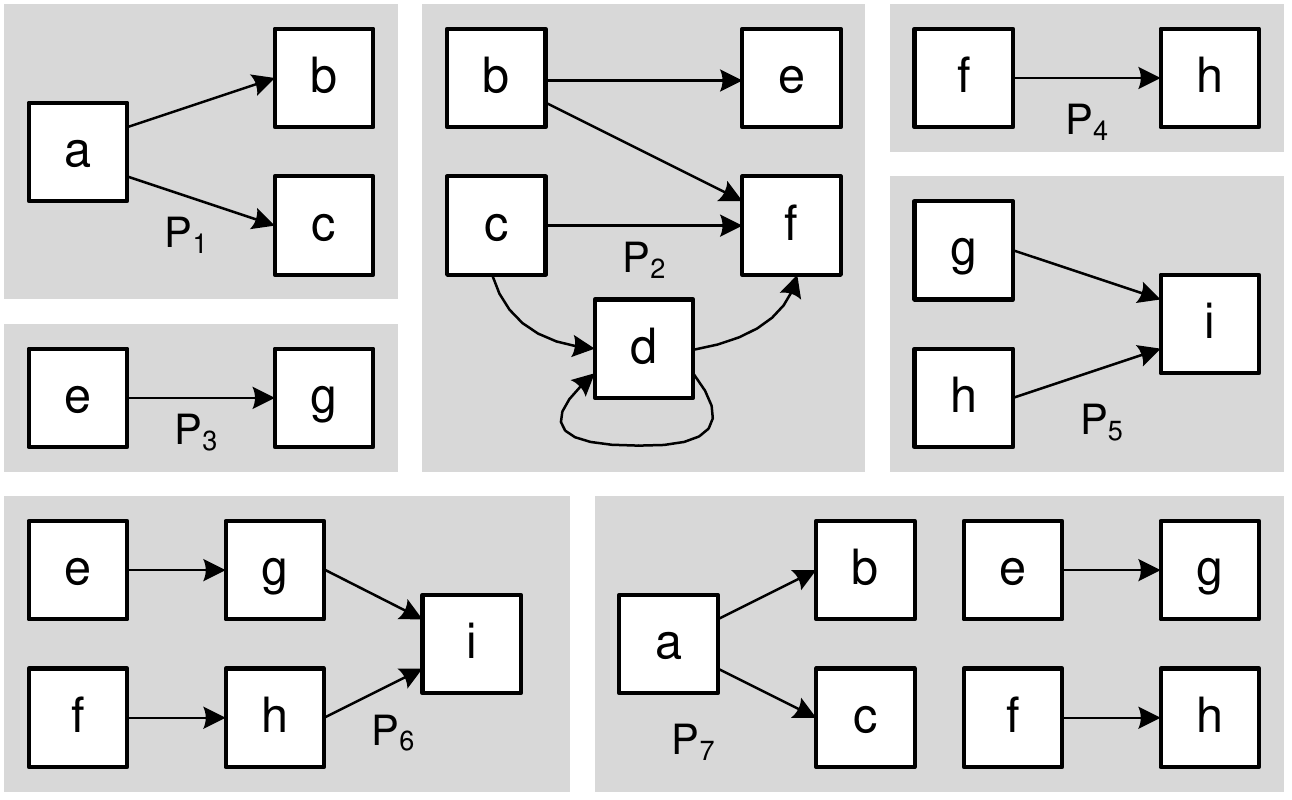}}
\caption{Seven example passages of graph $G_1$ shown in Figure~\ref{fig-passage-init}.}
\label{fig-passage-split}
\end{figure}
\begin{lemma}[Trivial Passages]
Let $G=(V,E)$ be a graph. The \emph{empty passage} $\emptyset$ and the \emph{full passage} $E$ are trivial passages of $G$.
Formally: $\{\emptyset, E\} \subseteq \mi{pas}(G)$ for any $G$.
\end{lemma}
Some of the passages in Figure~\ref{fig-passage-split} are overlapping:
$P_6 = P_3 \cup P_4 \cup P_5$ and $P_7 = P_1 \cup P_3 \cup P_4$.
To combine passages into new passages and to reason about the properties of passages we define the following notations.
\begin{definition}[Passage Operators]
Let $G=(V,E)$ be a graph with $P,P_1,P_2 \subseteq E$.
$P_1 \cup P_2$, $P_1 \cap P_2$, $P_1 \setminus P_2$, $P_1 = P_2$, $P_1 \neq P_2$,
$P_1 \subseteq P_2$, and $P_1 \subset P_2$ are defined as usual.
$\pi_1(P) = \{x \mid (x,y) \in P\}$ are the initial vertices of $P$,
$\pi_2(P) = \{y \mid (x,y) \in P\}$ are the terminal vertices of $P$,
$P_1 \# P_2$ if and only if $P_1 \cap P_2 = \emptyset$,
$P_1 \triangleright P_2$ if and only if $\pi_2(P_1) \cap \pi_1(P_2) \neq \emptyset$.
\end{definition}
Note that $d$ is both an initial and terminal vertex of $P_2$ in Figure~\ref{fig-passage-split}:
$\pi_1(P_2) = \{b,c,d\}$ and $\pi_2(P_2) = \{d,e,f\}$.
$P_5 \# P_7$ because $P_5 \cap P_7 = \emptyset$.
$P_4 \triangleright P_5$ because $\pi_2(P_4) \cap \pi_1(P_5) = \{ h \} \neq \emptyset$.

The union, intersection and difference of passages yield passages.
For example, $P_7 = P_1 \cup P_3 \cup P_4$ is a passage composed of three smaller passages.
$P_5 = P_6 \setminus P_7$ and $P_6 \cap P_7 = P_3 \cup P_4$  are passages.
\begin{lemma}[Passages Are Closed under $\cup$, $\cap$ and $\setminus$]\label{lem:closed}
Let $G=(V,E)$ be a graph. If $P_1,P_2 \in \mi{pas}(G)$ are two passages, then
$P_1 \cup P_2$, $P_1 \cap P_2$, and $P_1 \setminus P_2$ are also passages.
\end{lemma}
\begin{proof}
Let $P_1,P_2 \in \mi{pas}(G)$, $(x,y) \in P_1 \cup P_2$, and $\{(x,y'),\allowbreak (x',y)\} \subseteq E$. We need to show that $\{(x,y'),(x',y)\} \subseteq P_1 \cup P_2$.
    If $(x,y) \in P_1$, then $\{(x,y'),(x',y)\} \subseteq P_1 \subseteq P_1 \cup P_2$.
    If $(x,y) \in P_2$, then $\{(x,y'),(x',y)\} \subseteq P_2 \subseteq P_1 \cup P_2$.

Let $P_1,P_2 \in \mi{pas}(G)$, $(x,y) \in P_1 \cap P_2$, and $\{(x,y'),\allowbreak (x',y)\} \subseteq E$. We need to show that $\{(x,y'),(x',y)\} \subseteq P_1 \cap P_2$.
    Since $(x,y) \in P_1$, $\{(x,y'),(x',y)\} \subseteq P_1$.
    Since $(x,y) \in P_2$, $\{(x,y'),(x',y)\} \subseteq P_2$. Hence, $\{(x,y'),(x',y)\} \subseteq P_1 \cap P_2$.

Let $P_1,P_2 \in \mi{pas}(G)$, $(x,y) \in P_1 \setminus P_2$, and $\{(x,y'),(x',y)\} \subseteq E$. We need to show that $\{(x,y'),(x',y)\} \subseteq P_1 \setminus P_2$.
    Since $(x,y) \in P_1$, $\{(x,y'),(x',y)\} \subseteq P_1$.
    Since $(x,y) \not\in P_2$, $\{(x,y'),(x',y)\} \cap P_2 = \emptyset$. Hence, $\{(x,y'),(x',y)\} \subseteq P_1 \setminus P_2$. \qed
\end{proof}
A passage is fully characterized by both the set of initial vertices and the set of terminal vertices.
Therefore, the following properties hold.
\begin{lemma}[Passage Properties]
Let $G=(V,E)$ be a graph. For any $P_1,P_2 \in \mi{pas}(G)$:
\begin{itemize}
\item $\pi_1(P_1) = \pi_1(P_2) \ \Leftrightarrow \ P_1 = P_2  \ \Leftrightarrow \ \pi_2(P_1) = \pi_2(P_2)$,
\item $P_1 \# P_2 \ \Leftrightarrow \ \pi_1(P_1) \cap \pi_1(P_2) = \emptyset$, and
\item $P_1 \# P_2 \ \Leftrightarrow \ \pi_2(P_1) \cap \pi_2(P_2) = \emptyset$.
\end{itemize}
\end{lemma}
\begin{proof}
$X=\pi_1(P)$ implies $P = \{(x,y) \in E \mid x \in X\}$ (definition of passages).
Hence, $\pi_1(P_1) = \pi_1(P_2) \ \Rightarrow \ P_1 = P_2$ (because a passage $P$ is fully determined by $\pi_1(P)$). The other direction ($\Leftarrow$) holds trivially.
A passage $P$ is also fully determined by $\pi_2(P)$. Hence, $\pi_2(P_1) = \pi_2(P_2) \ \Rightarrow \ P_1 = P_2$. Again the other direction ($\Leftarrow$) holds trivially.

The second property follows from the observation that two passages share an edge if and only if the initial vertices overlap.
If two passages share an edge $(x,y)$, they also share initial vertex $x$.
If two passage share initial vertex $x$, then they also share some edges $(x,y)$.

Due to symmetry, the same holds for the third property.\qed
\end{proof}
The following lemma shows that a passage can be viewed as a fixpoint: $P=((\pi_1(P)\times V) \cup (V \times \pi_2(P)))\cap E$.
This property will be used to construct minimal passages.

\begin{lemma}[Another Passage Characterization]\label{lemcharalt}
Let $G=(V,E)$ be a graph. $P \subseteq E$ is a passage if and only if $P=((\pi_1(P)\times V) \cup (V \times \pi_2(P)))\cap E$.
\end{lemma}
\begin{proof}
Suppose $P$  is a passage: it is fully characterized by $\pi_1(P)$ and $\pi_2(P)$.
Take all edges leaving from $\pi_1(P)$: $P=(\pi_1(P)\times V)\cap E$.
Take all edges entering $\pi_2(P)$: $P=(V \times \pi_2(P))\cap E$.
Hence, $P= (\pi_1(P)\times V)\cap E = (V \times \pi_2(P))\cap E$. So, $P=((\pi_1(P)\times V) \cup (V \times \pi_2(P)))\cap E$.

Suppose $P=((\pi_1(P)\times V) \cup (V \times \pi_2(P)))\cap E$.
Let $(x,y)\in P$ and $\{(x,y'),(x',y)\} \subseteq E$.
Clearly, $(x,y')  \in (\pi_1(P)\times V)\cap E$ and $(x',y) \in (V \times \pi_2(P))\cap E$.
Hence, $\{(x,y'),(x',y)\} \subseteq ((\pi_1(P)\times V) \cup (V \times \pi_2(P)))\cap E = P$. \qed
\end{proof}

\section{Passage Partitioning}\label{sec:paspart}

After introducing the notion of passages and their properties, we now show that graph can be \emph{partitioned} using passages.
For example, the set of passages $\{P_1,P_2,P_3,P_4,\allowbreak P_5\}$ in Figure~\ref{fig-passage-split} partitions $G_1$.
Other passage partitionings for graph $G_1$ are $\{P_2,P_5,P_7\}$ and $\{P_1,P_2,P_6\}$.
\begin{definition}[Passage Partitioning]
Let $G=(V,E)$ be a graph.
${\cal P} = \{P_1,P_2, \ldots ,P_n\} \subseteq \mi{pas}(G)\setminus \{\emptyset\}$ is a \emph{passage partitioning} if and only if
$\bigcup {\cal P} = E$ and $\forall_{1 \leq i < j \leq n}\ \allowbreak P_i \# P_j$.
\end{definition}
Any passage partitioning ${\cal P}$ defines an equivalence relation on the set of edges.
For $e_1,e_2 \in E$, $e_1 \sim_{{\cal P}} e_2$
if there exists a $P\in {\cal P}$ with $\{e_1,e_2\} \subseteq P$.
\begin{lemma}[Equivalence Relation]\label{lem:minarg}
Let $G \allowbreak = (V,E)$ be a graph with passage partitioning ${\cal P}$.
$\sim_{\cal P}$ defines an equivalence relation.
\end{lemma}
\begin{proof}
We need to prove that $\sim_{\cal P}$ is reflexive, symmetric, and transitive.
Let $e,e',e'' \in E$.
Clearly, $e \sim_{\cal P} e$ because $e \in E = \bigcup {\cal P}$ (${\cal P}$ is a passage partitioning).
Hence, there must be a $P\in {\cal P}$ with $e \in {\cal P}$ (reflexivity).
If $e \sim_{\cal P} e'$, then $e' \sim_{\cal P} e$ (symmetry).
If $e \sim_{\cal P} e'$ and $e' \sim_{\cal P} e''$,
then there must be a $P\in {\cal P}$ such that $\{e_1,e_2,e_3\} \subseteq P$.
Hence, $e \sim_{\cal P} e''$ (transitivity).
\qed
\end{proof}
Any graph has a passage partitioning, e.g., $\{ E\}$ is always a valid passage partitioning.
However, to decompose analysis one is typically interested in partitioning the graph in as many passages as possible.
Therefore, we introduce the notion of a \emph{minimal} passage.
Passage $P_6$ in Figure~\ref{fig-passage-split} is not minimal because it contains smaller non-empty passages: $P_3$, $P_4$, and $P_5$.
Passage $P_7$ is also not minimal.
Only the first five passages in Figure~\ref{fig-passage-split} ($P_1$, $P_2$, $P_3$, $P_4$ and $P_5$) are minimal.
\begin{definition}[Minimal Passages]
Let $G=(V,E)$ be a graph and $P \in \mi{pas}(G)$ a passage.
$P$ is minimal if and only if there is no non-empty passage $P' \in \mi{pas}(G)\setminus \{\emptyset\}$ such that $P' \subset P$.
$\mi{pas}_{\mi{min}}(G)$ is the set of all non-empty minimal passages.
\end{definition}
Two different minimal passages cannot share the same edge. Otherwise, the difference between both passages would yield a smaller non-empty minimal passage.
Hence, an edge can be used to uniquely identify a minimal passage.
The fixpoint characterization given in Lemma~\ref{lemcharalt} suggests an iterative procedure that starts with a single edge.
In each iteration edges are added that must be part of the same minimal passage. As shown this procedure can be used to determine all minimal passages.

\begin{lemma}[Constructing Minimal Passages]\label{lem:minarg}
Let $G \allowbreak = (V,E)$ be a graph. For any $(x,y)\in E$, there exists precisely one minimal passage $P_{(x,y)} \in \mi{pas}_{\mi{min}}(G)$ such that $(x,y) \in P_{(x,y)}$.
\end{lemma}
\begin{proof}
Initially, set $P:=\{(x,y)\}$.
Extend $P$ as follows: $P := ((\pi_1(P)\times V) \cup (V \times \pi_2(P)))\cap E$.
Repeat extending $P$ until it does not change anymore.
Finally, return $P_{(x,y)} =P$.
The procedure ends because the number of edges is finite.
If $P = ((\pi_1(P)\times V) \cup (V \times \pi_2(P)))\cap E$ (i.e., $P$ does not change anymore),
then $P$ is indeed a passage (see Lemma~\ref{lemcharalt}).
$P$ is minimal because no unnecessary edges are added:
if $(x,y) \in P$, then any edge starting in $x$ or ending in $y$ has to be included.

To prove the latter one can also consider all passages
${\cal P} = \{P_1, P_2,  \ldots, P_n\}$ that contain $(x,y)$.
The intersection of all such passages $\bigcap {\cal P}$ contains
edge $(x,y)$ and is again a passage because
of Lemma~\ref{lem:closed}. Hence, $\bigcap {\cal P} = P_{(x,y)}$.
\qed
\end{proof}
The construction described in the proof can be used compute all minimal passages
and is quadratic in the number of edges.

$\mi{pas}_{\mi{min}}(G_1) = \{P_1,P_2,P_3,P_4,P_5\}$ for the graph shown in Figure~\ref{fig-passage-init}.
This is also a passage partitioning. (Note that the construction in Lemma~\ref{lem:minarg} is similar to the computation
of so-called clusters in a Petri net \cite{deselesparza}.)

\begin{theorem}[Minimal Passage Partitioning]
Let $G \allowbreak = \allowbreak (V,E)$ be a graph. $\mi{pas}_{\mi{min}}(G)$ is a passage partitioning.
\end{theorem}
\begin{proof}
Let $\mi{pas}_{\mi{min}}(G) = \{ P_1, P_2, \ldots, P_n\}$.
Clearly, $\{ P_1, \allowbreak P_2, \ldots, P_n\} \subseteq \mi{pas}(G)\setminus \{\emptyset\}$,
$\bigcup_{1 \leq i \leq n} P_i = E$ and $\forall_{1 \leq i < j \leq n}\ \allowbreak P_i \# P_j$ (follows from Lemma~\ref{lem:minarg}). \qed
\end{proof}
Figure~\ref{fig-example-part} shows a larger graph $G_2 = (V_2,E_2)$ with $V_2=\{a,b, \ldots ,o\}$ and $E_2=\{(a,b),(b,e), \ldots ,(n,o)\}$.
The figure also shows six passages. These form a passage partitioning.
Each edge has a number that refers to the corresponding passage, e.g., edge $(h,k)$ is part of passage $P_4$.
Passages are shown as rectangles and vertices are put on the boundaries of at most two passages.
Vertex $a$ in Figure~\ref{fig-example-part} is on the boundary of $P_1$ because $(a,b) \in P_1$.
Vertex $b$ is on the boundary of $P_1$ and $P_2$ because $(a,b) \in P_1$ and $(b,e) \in P_2$.
$G_2$ has no isolated vertices, so all vertices are on the boundary of at least one passage.
\begin{figure}[htb]
\centerline{\includegraphics[width=10.5cm]{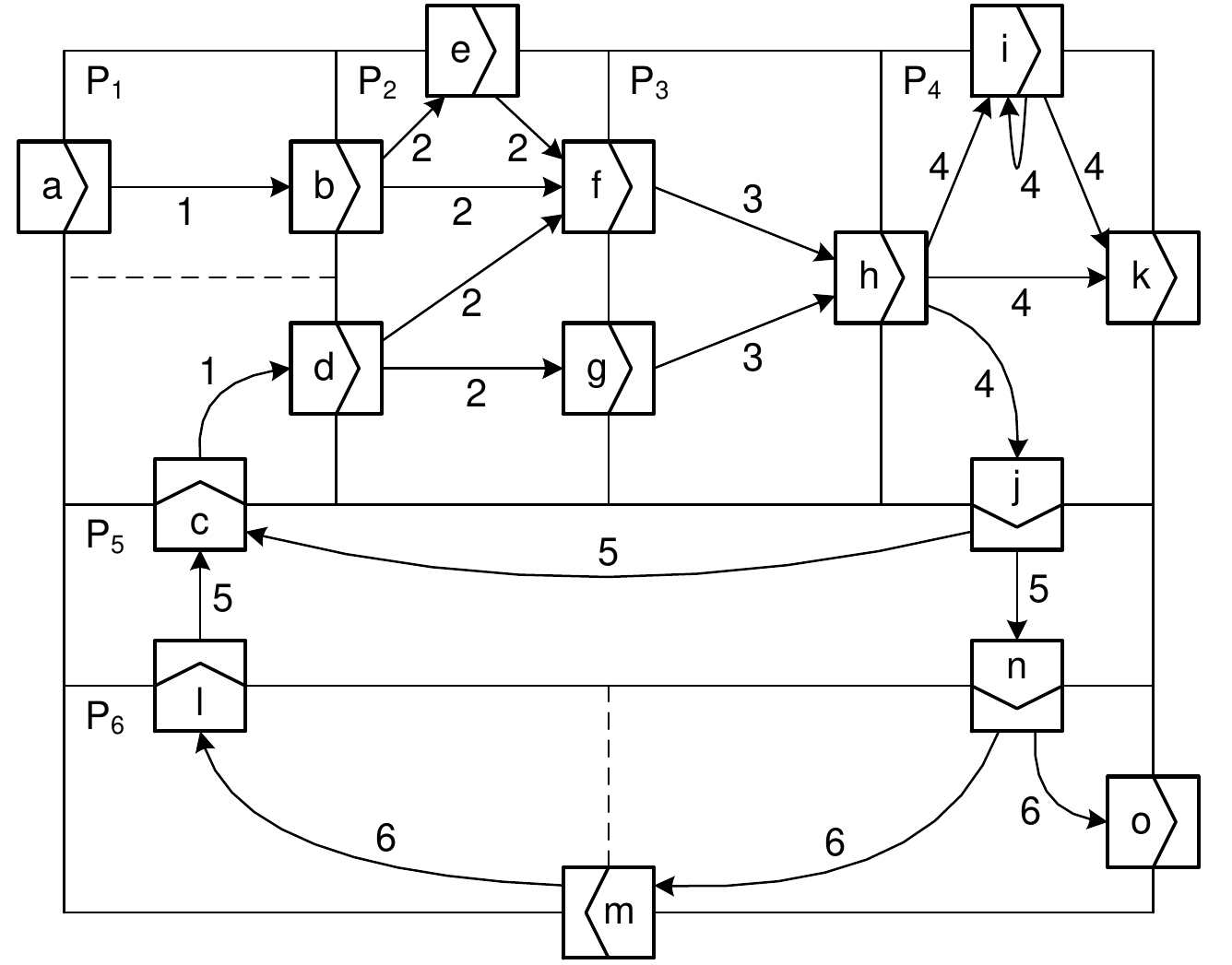}}
\caption{A passage partitioning for graph $G_2$.}
\label{fig-example-part}
\end{figure}

The passage partitioning shown in Figure~\ref{fig-example-part} is not composed of minimal passages
as is indicated by the two dashed lines. Both $P_1$ and $P_6$ are not minimal.
$P_1$ can be split into minimal passages $P_{1a} = \{(a,b)\}$ and  $P_{1b} = \{(c,d)\}$.
$P_6$ can be split into minimal passages $P_{6a} = \{(m,l)\}$ and  $P_{6b} = \{(n,o),(n,m)\}$.
In fact, as shown next, any passage can be decomposed into minimal non-empty passages.

\begin{theorem}[Composing Minimal Passages]
Let $G=(V,E)$ be a graph. For any passage $P\in \mi{pas}(G)$ there is a set of minimal non-empty passages $\{ P_1,P_2, \ldots, P_n\} \subseteq \mi{pas}_{\mi{min}}(G)$
such that $\bigcup_{1 \leq i \leq n} P_i = P$ and $\forall_{1 \leq i < j \leq n}\ P_i \# P_j$.
\end{theorem}
\begin{proof}
Let $\{ P_1,P_2, \ldots, P_n\} = \{ P_{(x,y)} \mid (x,y) \in P\}$.
These passages are minimal (Lemma~\ref{lem:minarg}) and also cover all edges in $P$.
Moreover, two different minimal passages cannot share edges. \qed
\end{proof}
A graph without edges has only one passage.
Hence, if $E = \emptyset$, then $\mi{pas}(G) = \{ \emptyset \}$ (just one passage), $\mi{pas}_{\mi{min}}(G) = \emptyset$ (no minimal non-empty passages), and $\emptyset$ is the only passage partitioning.
If $E \neq \emptyset$, then there is always a trivial singleton passage partitioning $\{ E\}$ and a minimal passage partitioning $\mi{pas}_{\mi{min}}(G)$ (but there may be many more).

\begin{lemma}[Number of Passages]
Let $G=(V,E)$ be a graph with $k=|\mi{pas}_{\mi{min}}(G)|$ minimal non-empty passages.
There are $2^k$ passages and $B_k$ passage partitionings.\footnote{$B_k$ is the $k$-th Bell number (the number of partitions of a set of size $k$), e.g., $B_3=5$, $B_4=15$, and $B_5=52$ \cite{bell-numbers}.}
For any passage partitioning $\{ P_1,P_2, \ldots, P_n\}$ of $G$: $n \leq k \leq |E|$.
\end{lemma}
\begin{proof}
Any passage can be composed of minimal non-empty passages. Hence, there are $2^k$ passages.
$B_k$ is the number of partitions of a set with $k$ members, thus corresponding to the number of passage partitionings.

If there are no edges, there are no minimal non-empty passages ($k=0$) and there is only one possible passage partitioning: $\emptyset$. Hence, $n=0$.
If $E \neq \emptyset$, then $\mi{pas}_{\mi{min}}(G)$ is the most refined passage partitioning.
There are at most $|E|$ minimal non-empty passages as they cannot share edges. Hence, $n \leq k \leq |E|$.
Note that $n \geq 1$ if $E \neq \emptyset$.
\qed
\end{proof}
Graph $G_2$ in Figure~\ref{fig-example-part} has $2^{8} = 256$ passages and $B_8 = 4140$ passage partitionings.

\section{Passage Graphs}\label{sec:pasgraph}

Passage partitionings can be visualized using \emph{passage graphs}.
To relate passages, we first define the input/output vertices of a passage.

\begin{definition}[Input and Output Vertices]
Let $G=(V,E)$ be a graph and $P \in \mi{pas}(G)$ a passage.
$\mi{in}(P) = \pi_1(P) \setminus \pi_2(P)$ are the input vertices of $P$,
$\mi{out}(P) = \pi_2(P) \setminus \pi_1(P)$ are the output vertices of $P$, and
$\mi{io}(P) = \pi_1(P) \cap \pi_2(P)$ are the input/output vertices of $P$.
\end{definition}
Note the difference between input, output, and input/output vertices on the one hand
and the initial and terminal vertices of a passage on the other hand.
Given a passage partitioning, there are five types of vertices:
isolated vertices, input vertices, output vertices, connecting vertices, and local vertices.

\begin{definition}[Five Types of Vertices]
Let $G=(V,E)$ be a graph and ${\cal P} = \{ P_1,P_2, \ldots, P_n\}$ a passage partitioning.
$V_{\mi{iso}} = V \setminus (\pi_1(E) \cup \pi_2(E))$ are the isolated vertices of ${\cal P}$,
$V_{\mi{in}} = \pi_1(E) \setminus \pi_2(E)$ are the input vertices of ${\cal P}$,
$V_{\mi{out}} = \pi_2(E) \setminus \pi_1(E)$ are the output vertices of ${\cal P}$,
$V_{\mi{con}} = \bigcup_{i \neq j} \pi_2(P_i) \cap \pi_1(P_j)$ are the connecting vertices of ${\cal P}$,
$V_{\mi{loc}} = \bigcup_{i} \pi_1(P_i) \cap \pi_2(P_i)$ are the local vertices of ${\cal P}$.
\end{definition}
Note that $V =
V_{\mi{iso}} \cup
V_{\mi{in}} \cup
V_{\mi{out}} \cup
V_{\mi{con}} \cup
V_{\mi{loc}}$ and the five sets are pairwise disjoint, i.e., they partition $V$.
In the passage partitioning shown in Figure~\ref{fig-example-part}:
$a$ is the only input vertex,
$k$ and $o$ are output vertices, and
$e$, $i$ and $m$ are local vertices.
All other vertices are connecting vertices.

\begin{definition}[Passage Graph]
Let $G=(V,E)$ be a graph and ${\cal P} = \{ P_1,P_2, \ldots, P_n\}$ a passage partitioning.
$({\cal P},\{(P,P')\in {\cal P} \times {\cal P} \mid P \triangleright P' \})$ is corresponding passage graph .
\end{definition}
Figure~\ref{fig-example-pass-graph} shows a passage graph.
The graph shows the relationships among passages and can be used to partition the vertices $V$ into
$V_{\mi{iso}} \cup
V_{\mi{in}} \cup
V_{\mi{out}} \cup
V_{\mi{con}} \cup
V_{\mi{loc}}$.
\begin{figure}[htb]
\centerline{\includegraphics[width=8cm]{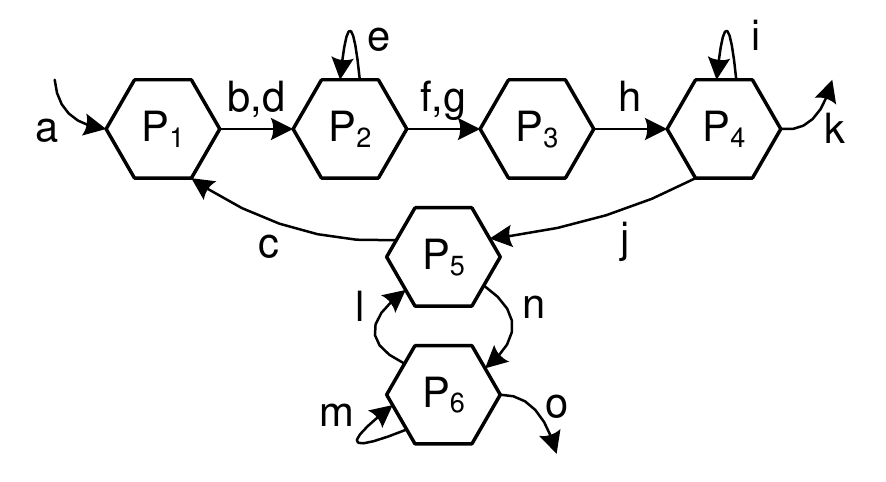}}
\caption{Passage graph based on the passage partitioning shown in Figure~\ref{fig-example-part}.}
\label{fig-example-pass-graph}
\end{figure}

\section{Quality of a Passage Partitioning}\label{sec:quality}

Passages can be used to decompose analysis problems (e.g., conformance checking and process discovery  \cite{aalst-passages-PN2012}).
In the extreme case, there is just one minimal passage covering all edges in the graph. In this case, the graph cannot be decomposed.
Ideally, we would like to use a passage partitioning ${\cal P} = \{ P_1,P_2, \ldots, P_n\}$ that is accurate and that has only small passages.
One could aim at as many passages as possible in order to minimize the average size per passage: $\mi{av}({\cal P}) = \frac{|E|}{n}$ per passage.
One can also aim at minimizing the size of the biggest passage (i.e., $\mi{big}({\cal P}) = \mi{max}_{1 \leq i \leq n}\ |P_i|$) because the biggest passage often takes most of the computation time.

To have smaller passages, one may need to abstract from edges that are less important. To reason about such ``approximate passages'' we define
the input as $G_{\pi} = (V,\pi)$ with vertices $V$ and weight function $\pi \in (V \times V) \rightarrow [-1,1]$.
Given two vertices $x,y \in V$: $\pi(x,y)$ is ``weight'' of the possible edge connecting $x$ and $y$.
If  $\pi(x,y) > 0$, then it is more likely than unlikely that there is an edge connecting $x$ and $y$.
If  $\pi(x,y) < 0$, then it is more unlikely than likely that there is an edge connecting $x$ and $y$.
One can view $\frac{\pi(x,y)+1}{2}$ as the ``probability'' that there is such an edge.
The penalty for leaving out an edge $(x,y)$ with $\pi(x,y) = 0.99$ is much bigger than leaving out an edge $(x',y')$ with $\pi(x',y') = 0.15$.
The accuracy of a passage partitioning ${\cal P} = \{ P_1,P_2, \ldots, P_n\}$ with $E = \cup_{1 \leq i \leq n}\ P_i$ for input $G_{\pi} = (V,\pi)$
can be defined as $\mi{acc}({\cal P}) = \frac{\sum_{(x,y)\in E}  \pi(x,y)}{\mi{max}_{E' \subseteq V \times V} \sum_{(x,y)\in E'}  \pi(x,y)}$.
If $\mi{acc}({\cal P}) = 1$, then all edges having a positive weight are included in some passage and none of edges having a negative weight are included.
Often there is a trade-off between higher accuracy and smaller passages, e.g.,
discarding a potential edge having a low weight may allow for splitting a large passage into two smaller ones.
Just like in traditional graph partitioning \cite{karpis_graph-partitioning1998,kernighan_graph-partitioning1970}, one can look for the passage partitioning that
maximizes $\mi{acc}({\cal P})$ provided that $\mi{av}({\cal P}) \leq \tau_{\mi{av}}$ and/or
$\mi{big}({\cal P}) \leq \tau_{\mi{big}}$, where $\tau_{\mi{av}}$ and $\tau_{\mi{big}}$ are suitably chosen thresholds.
Whether one needs to resort to approximate passages depends on the domain, e.g., when discovering process models from event logs
causalities tend to be uncertain and including all potential causalities results in Spaghetti-like graphs \cite{process-mining-book-2011}, therefore approximate passages are quite useful.

\section{Conclusion}\label{sec:concl}

In this article we introduced the new notion of passages. Passages have been shown to be useful in the domain of process mining.
Given their properties and possible applications in other domains, we examined passages in detail.
Passages are closed under the standard set operators (union, difference, and intersection).
A graph can be partitioned into components based on its minimal passages and any passage is composed of minimal passages.
The theory of passages can be extended to deal with approximate passages. We plan to examine these in the context of process mining,
but are also looking for applications of passage partitionings in other domains (e.g., distributed enactment and verification).

%\section*{References}
%\bibliographystyle{elsarticle-num}
\bibliographystyle{plain}
\bibliography{../../bib/lit}

\begin{thebibliography}{10}

\bibitem{process-mining-book-2011}
{W.M.P. van der} Aalst.
\newblock {\em {Process Mining: Discovery, Conformance and Enhancement of
  Business Processes}}.
\newblock Springer-Verlag, Berlin, 2011.

\bibitem{aalst-passages-PN2012}
{W.M.P. van der} Aalst.
\newblock {Decomposing Process Mining Problems Using Passages}.
\newblock In S.~Haddad and L.~Pomello, editors, {\em {Applications and Theory
  of Petri Nets 2012}}, volume 7347 of {\em Lecture Notes in Computer Science},
  pages 72--91. Springer-Verlag, Berlin, 2012.

\bibitem{digraphs-book-2004}
J.~Bang-Jensen and G.~Gutin.
\newblock {\em {Digraphs: Theory, Algorithms and Applications (Second
  Edition)}}.
\newblock Springer-Verlag, Berlin, 2009.

\bibitem{deselesparza}
J.~Desel and J.~Esparza.
\newblock {\em {Free Choice Petri Nets}}, volume~40 of {\em Cambridge Tracts in
  Theoretical Computer Science}.
\newblock Cambridge University Press, Cambridge, UK, 1995.

\bibitem{feige-vertex-graph-partitioning2005}
U.~Feige, M.~Hajiaghayi, and J.~Lee.
\newblock {Improved Approximation Algorithms for Minimum-Weight Vertex
  Separators}.
\newblock In {\em Proceedings of the thirty-seventh annual ACM symposium on
  Theory of computing}, pages 563--572. ACM, New York, 2005.

\bibitem{handbook_graph_theory2004}
J.L. Gross and J.~Yellen.
\newblock {\em {Handbook of Graph Theory}}.
\newblock CRC Press, 2004.

\bibitem{karpis_graph-partitioning1998}
G.~Karpis and V.~Kumar.
\newblock {A Fast and High Quality Multilevel Scheme for Partitioning Irregular
  Graphs}.
\newblock {\em SIAM Journal on Scientific Computing}, 20(1):359--392, 1998.

\bibitem{kernighan_graph-partitioning1970}
B.W. Kernighan and S.~Lin.
\newblock {An Efficient Heuristic Procedure for Partitioning Graphs}.
\newblock {\em The Bell Systems Technical Journal}, 49(2), 1970.

\bibitem{kim-vertex-graph-partitioning2012}
M.~Kim and K.~Candan.
\newblock {SBV-Cut: Vertex-Cut Based Graph Partitioning Using Structural
  Balance Vertices}.
\newblock {\em Data and Knowledge Engineering}, 72:285--303, 2012.

\bibitem{bell-numbers}
N.J.A. Sloane.
\newblock {Bell Numbers}.
\newblock In {\em {Encyclopedia of Mathematics}}. Kluwer Academic Publishers,
  2002.
\newblock
  {\url{http://www.encyclopediaofmath.org/index.php?title=Bell_numbers&oldid=14335}}.

\end{thebibliography}






% \end{thebibliography}

\end{document}